\documentclass[12pt]{article}
\usepackage{amssymb}
\usepackage{amsmath}
\usepackage{amsthm}
\usepackage{multirow}

\usepackage[utf8]{inputenc}

\usepackage[normalem]{ulem}

\numberwithin{equation}{section}
\newcommand{\be}{\begin{eqnarray}}
\newcommand{\ee}{\end{eqnarray}}
\newcommand{\non}{\nonumber}
\newcommand{\id}{\mathbb{I}}
\newcommand{\tr}{\mathop{\rm tr}\nolimits}

\newcommand{\diag}{\mathop{\rm diag}\nolimits}

\newcommand{\mA}{\mathcal{A}}
\newcommand{\mB}{\mathcal{B}}
\newcommand{\mC}{\mathcal{C}}
\newcommand{\mD}{\mathcal{D}}

\newcommand{\mQ}{\mathcal{Q}}

\newcommand{\refs}{|0\rangle}
\newcommand{\drefs}{\langle0|}

\usepackage[usenames,dvipsnames]{color}

\newtheorem*{prop*}{Proposition}

\newtheorem{lemma}{Lemma}

\newcommand\blfootnote[1]{%
  \begingroup
  \renewcommand\thefootnote{}\footnote{#1}%
  \addtocounter{footnote}{-1}%
  \endgroup
}

\begin{document}

\begin{titlepage}
\strut\hfill UMTG--287
\vspace{.5in}
\begin{center}

\LARGE 
Universal Bethe ansatz solution\\[0.2in]
for the Temperley-Lieb spin chain\\
\vspace{1in}
\large 
Rafael I. Nepomechie \footnote{
Physics Department,
P.O. Box 248046, University of Miami, Coral Gables, FL 33124 USA}
and Rodrigo A. Pimenta ${}^{1,}$\footnote{
Departamento de F\'{i}sica, Universidade Federal de S\~{a}o Carlos, Caixa Postal 676, 
CEP 13565-905, S\~{a}o Carlos, Brasil}\\[0.8in]
\end{center}

\vspace{.5in}

\begin{abstract}
We consider the Temperley-Lieb (TL) open quantum spin chain
with ``free'' boundary conditions associated with the
spin-$s$ representation of quantum-deformed $sl(2)$.  We
construct the transfer matrix, and determine its eigenvalues and the
corresponding Bethe equations using analytical Bethe ansatz.  We show
that the transfer matrix has quantum group symmetry, and we propose
explicit formulas for the number of solutions of the Bethe equations
and the degeneracies of the transfer-matrix eigenvalues.  We propose
an algebraic Bethe ansatz construction of the off-shell Bethe states,
and we conjecture that the on-shell Bethe states are highest-weight
states of the quantum group.  We also propose a determinant formula
for the scalar product between an off-shell Bethe state and its
on-shell dual, as well as for the square of the norm.  We find that
all of these results, except for the degeneracies and a constant
factor in the scalar product, are universal in
the sense that they do not depend on the value of the spin.  In an
appendix, we briefly consider the closed TL spin chain
with periodic boundary conditions, and show how a previously-proposed
solution can be improved so as to obtain the complete
(albeit non-universal) spectrum.
\end{abstract}

\blfootnote{e-mail addresses: {\tt nepomechie@physics.miami.edu, pimenta@df.ufscar.br}}

\end{titlepage}

\setcounter{footnote}{0}

\section{Introduction}

The generators $\{X_{(1)}, \ldots, X_{(N-1)} \}$ 
of the unital Temperley-Lieb (TL) algebra $TL_{N}$ \cite{Temperley:1971iq},
\be
X^{2}_{(i)} &=& c X_{(i)} \,, \non \\
X_{(i)} X_{(i\pm 1)} X_{(i)} &=& X_{(i)}\,,  \non \\
X_{(i)} X_{(j)} &=& X_{(i)} X_{(j)}\,, \qquad |i-j| > 1 \,,
\label{TLalgebra}
\ee
can be used to define the Hamiltonian of an open quantum spin chain
of length $N$ with ``free'' boundary conditions
\be
H=\sum_{i=1}^{N-1}X_{(i)}\,.
\label{Hamiltonian}
\ee
This type of model has been the subject of many investigations. 
For simplicity, we focus here on the models associated with 
$U_{Q}(A_{1}) = U_{Q}sl(2)$. The generators 
$X_{(i)}$ have been constructed for any value of spin $s$ 
\cite{Batchelor:1989uk, Batchelor:1989iz}. The $s=1/2$ case
is the well-known quantum-group-invariant spin-1/2 XXZ chain 
\cite{Pasquier:1989kd}. The $s=1$ case is the quantum deformation \cite{Batchelor:1989uk} of the
pure biquadratic spin-1 chain \cite{Parkinson1, Parkinson2, Barber:1989zz}.
These models are integrable; and closed-chain versions 
with periodic boundary conditions
have been investigated for $s>1/2$ 
using inversion relations \cite{Kluemper1, Kluemper2},
numerically \cite{Alcaraz:1992uq}, 
by coordinate Bethe ansatz \cite{Parkinson1, Parkinson2, Koberle:1993in}, 
and by analytical Bethe ansatz \cite{Kulish}. Additional 
results can be found in \cite{Martin:1994,
Doikou:2005pn, Aufgebauer:2010gg, deGier:2003iu, Ribeiro:2013} for the open chain, 
and in \cite{Levy:1991, Martin:1992td, Martin:1993jka} for the closed chain.
TL models associated with higher-rank algebras have also been 
investigated \cite{Batchelor:1991uy, Koberle:1995sw, Ghiotto:2000}.

Despite these and further efforts, a number of fundamental problems
related to these models, such as the formulation of an algebraic Bethe
ansatz solution, have remained unsolved.  Moreover, the analytical
Bethe ansatz solution proposed in \cite{Kulish} does not give the
complete spectrum.

The goal of this paper is to address some of these problems.  We
construct the transfer matrix corresponding to the Hamiltonian
(\ref{Hamiltonian}), and we determine its eigenvalues using analytical
Bethe ansatz.  We prove that the transfer matrix has quantum group
symmetry, which accounts for the degeneracies of the spectrum.
We propose an algebraic Bethe ansatz construction of the Bethe states,
which (when on-shell) we conjecture are highest-weight states of the 
quantum group. The scalar product
between an off-shell Bethe state and an on-shell Bethe state is also considered, and we conjecture that
it can be given in terms of a determinant formula; the square of the norm, {\it i.e.}, the scalar product between on-shell Bethe
states, follows as a limit\footnote{ 
Such formulas are generally known as Slavnov \cite{Sla89} and Gaudin-Korepin \cite{Gaudin,PhysRevD.23.417,Korepin:1982gg} formulas,
respectively.}.
We find that all of these results, except for the degeneracies and a constant
factor in the scalar product, are 
universal in the sense that they do not depend on the value of
the spin. 

Although most of this paper concerns the open TL chain, we briefly
consider the closed TL chain
with periodic boundary conditions in an appendix.  There we revisit the
analytical Bethe ansatz computation in \cite{Kulish}, and show how 
the proposed solution can be improved so as to obtain the complete spectrum.
In contrast with the case of the open chain, the 
solutions of the closed-chain Bethe equations are not universal, as the Bethe roots 
depend on the value of the spin.
 
The outline of this paper is as follows.  In section
\ref{sec:transfer} we describe the construction of the Hamiltonian
(\ref{Hamiltonian}) and the corresponding transfer matrix.  In section
\ref{sec:BA} we use analytical Bethe ansatz to determine the eigenvalues
of the transfer matrix and the corresponding Bethe equations.  In
section \ref{sec:symmetry} we show that the transfer matrix has
quantum group symmetry, and we propose explicit formulas for the
number of solutions of the Bethe equations and the degeneracies of the
transfer-matrix eigenvalues.  In section \ref{sec:ABA} we present our 
proposals for the algebraic Bethe ansatz solution and scalar products.
We briefly discuss these results and remaining problems in section 
\ref{sec:discussion}. We treat the closed TL chain in 
Appendix \ref{sec:closed}.

\section{Transfer matrix}\label{sec:transfer}

We begin this section by describing in more detail the construction of the 
Hamiltonian (\ref{Hamiltonian}). We then construct the corresponding transfer 
matrix, which is the generating function of the Hamiltonian and the
higher local conserved commuting quantities, and we review some of its important 
properties.

We consider the TL open quantum spin chain corresponding to the
spin-$s$ representation of $U_{Q}sl(2)$.  The $X_{(i)}$ appearing in
the Hamiltonian (\ref{Hamiltonian}) are operators on $\left({\mathbb C}^{2s+1}\right)^{\otimes N}$ defined by
\be
X_{(i)} = X_{i, i+1} \,,
\ee
where $X$ is a $(2s+1)^{2}$ by $(2s+1)^{2}$ matrix (an 
endomorphism of ${\mathbb C}^{2s+1} \otimes {\mathbb C}^{2s+1}$) with the 
following matrix elements \cite{Batchelor:1989uk}
\be
\langle m_{1}, m_{2}|X|m'_{1}, m'_{2}\rangle & = &
(-1)^{m_{1}-m'_{1}} 
Q^{m_{1}+m'_{1}}\delta_{m_{1}+m_{2},0}\delta_{m'_{1}+m'_{2},0}\,, 
\ee
where $m_{1}, m_{2}, m'_{1}, m'_{2} = -s, -s+1, \ldots, s$,
and $s=\frac{1}{2}, 1, \frac{3}{2}, \ldots$.  In other words,
$X_{(i)}$ is an operator on $N$ copies of ${\mathbb C}^{2s+1}$, which acts as $X$ on 
copies $i$ and $i+1$, and otherwise as the identity operator, 
\be
X_{(i)} = \id^{\otimes(i-1)} \otimes X \otimes \id^{\otimes(N-i-1)} \,, 
\ee
where $\id$ is the identity operator on ${\mathbb C}^{2s+1}$. 
These operators satisfy the TL algebra (\ref{TLalgebra}), 
where $c$ is given by
\be
c = \left[2s+1\right]_{Q} = \frac{Q^{2s+1}-Q^{-2s-1}}{Q-Q^{-1}} =
 \sum_{k=-s}^{s} Q^{2k} \equiv -\left(q + 
\frac{1}{q}\right) \,.
\ee
We assume throughout this paper that $Q$ has a generic value. 

The Hamiltonian (\ref{Hamiltonian}) is integrable for any value of 
spin $s$. In the 
notation of \cite{Kulish}, the corresponding R-matrix is given by 
\cite{Jones:1990hq}
\be
R(u) = \left(u q - \frac{1}{u q}\right) {\cal P} + \left(u  - 
\frac{1}{u}\right) {\cal P} X \,,
\label{Rmatrix}
\ee
where ${\cal P}$ is the permutation matrix on ${\mathbb C}^{2s+1} 
\otimes {\mathbb C}^{2s+1}$. Indeed, the Yang-Baxter equation
\be
R_{12}(u_{1}/u_{2})\, R_{13}(u_{1}/u_{3})\, R_{23}(u_{2}/u_{3}) = 
R_{23}(u_{2}/u_{3})\, R_{13}(u_{1}/u_{3})\, R_{12}(u_{1}/u_{2}) 
\ee
is satisfied. This R-matrix has the unitarity property
\be
R_{12}(u) R_{21}(u^{-1}) = \zeta(u)\, \id^{\otimes 2} \,,
\qquad \zeta(u) = \omega(u q^{-1})\, \omega(u^{-1} q^{-1})\,,
\label{unitarity}
\ee
where $R_{21} = {\cal P}_{12}\, R_{12}\, {\cal P}_{12} = 
R_{12}^{t_{1} t_{2}}$, and $\omega(u)$ is defined as
\be
\omega(u) = u - \frac{1}{u} \,.
\ee 
This R-matrix also has crossing symmetry
\be
R_{12}(u) = V_{1} R^{t_{2}}_{12}(-u^{-1} q^{-1}) V_{1} \,,
\label{crossingR}
\ee 
where $V$ is an anti-diagonal matrix with elements
\be
V_{jk} = (-1)^{j} Q^{s+1-j}\delta_{j+k,2s+2} \,.
\ee

The model (\ref{Hamiltonian}) is an {\em open} spin chain.
For an integrable open spin chain, the transfer matrix is given by \cite{Sklyanin:1988yz} 
\be
t(u) = \tr_{0} K^{+}_{0}(u)\, T_{0}(u)\, K^{-}_{0}(u)\, \hat T_{0}(u)\,,
\label{transfergen}
\ee
where $T_{0}(u)$ and $\hat T_{0}(u)$ are the monodromy matrices
\be
T_{0}(u) = R_{0N}(u) \cdots R_{01}(u)\,, \qquad
\hat T_{0}(u) = R_{10}(u) \cdots R_{N0}(u)\,.
\label{monodromy}
\ee
Moreover, the K-matrices (endomorphisms of ${\mathbb C}^{2s+1}$) satisfy the 
boundary Yang-Baxter equations
\be
R_{12}(u/v)\, K^{-}_{1}(u)\, R_{21}(u v)\, K^{-}_{2}(v) = K^{-}_{2}(v)\, 
R_{12}(u v)\, K^{-}_{1}(u)\, R_{21}(u/v)
\ee
and \cite{Mezincescu:1990ui}
\be
\lefteqn{R_{12}(v/u)\, K^{+\ t_{1}}_{1}(u)\, M^{-1}_{1}\, R_{21}(u^{-1} 
v^{-1} q^{-2})\, M_{1}\, K^{+\ t_{2}}_{2}(v)}\non \\
&& = K^{+\ t_{2}}_{2}(v)\, 
M_{1}\, R_{12}(u^{-1} v^{-1} q^{-2})\, M^{-1}_{1}\, 
K^{+\ t_{1}}_{1}(u)\, R_{21}(v/u) \,,
\ee 
where $M$ is the diagonal matrix given by
\be
M=V^{t}\, V = \diag(Q^{-2s}\,, Q^{-2(s-1)}\,, \ldots\,, Q^{2s}) \,.
\ee
The Hamiltonian (\ref{Hamiltonian}) corresponds to the special case
with quantum-group invariance  \cite{Mezincescu:1990ui}
\be
K^{-} = \id\,, \qquad K^{+} = M \,,
\ee
and therefore the transfer matrix (\ref{transfergen}) takes the simpler form
\be
t(u) = \tr_{0} M_{0}\, T_{0}(u)\, \hat T_{0}(u)\,.
\label{transfer}
\ee
Indeed, the Hamiltonian is related to the transfer matrix as follows
\be
H = \alpha \frac{d}{du}t(u)\Big\vert_{u=1} + \beta\, \id^{\otimes N} \,,
\label{Htransfrltn}
\ee
where
\be
\alpha = -\left[4 \omega(q^{2})\, \omega(q)^{2N-2} \right]^{-1}\,, \qquad 
\beta = \frac{\omega(q)}{\omega(q^{2})}-\frac{N}{2}\frac{\omega(q^{2})}{\omega(q)}\,.
\label{alphabeta}
\ee
The higher conserved quantities can be obtained by
taking higher derivatives of the transfer matrix. These quantities 
commute with each other by virtue of the commutativity property \cite{Sklyanin:1988yz} 
\be
\left[ t(u)\,, t(v) \right] = 0 \,.
\label{commutativity}
\ee
The transfer matrix also has crossing symmetry \cite{Mezincescu:1991ag}
\be
t(u) = t(-u^{-1} q^{-1}) \,.
\label{crossingtransf}
\ee

\section{Analytical Bethe ansatz}\label{sec:BA}

We now proceed to determine the eigenvalues of the transfer matrix
(\ref{transfer}) by analytical Bethe ansatz \cite{Mezincescu:1991ag,
Reshetikhin:1983, Reshetikhin:1983vw, Wang2015}.  To this end, it is convenient
to introduce inhomogeneities $\{\theta_{j}\}$, i.e. to
consider instead the inhomogeneous transfer matrix
\be
t(u; \{\theta_{j}\}) = \tr_{0} M_{0}\, T_{0}(u; \{\theta_{j}\})\, 
\hat T_{0}(u; \{\theta_{j}\})\,,
\label{transferinhomo}
\ee
where
\be
T_{0}(u; \{\theta_{j}\}) = R_{0N}(u/\theta_{N}) \cdots 
R_{01}(u/\theta_{1})\,, \qquad
\hat T_{0}(u; \{\theta_{j}\}) = R_{10}(u \theta_{1}) \cdots R_{N0}(u \theta_{N})\,.
\label{monodromyinhomo}
\ee
As noted in \cite{Kulish}, the R-matrix (\ref{Rmatrix}) degenerates 
into a one-dimensional projector at $u=q^{-1}$,
\be
R(q^{-1}) = (q^{-1} - q) (2s+1) (-1)^{2s} P^{-}\,, \qquad  P^{-} = 
\frac{(-1)^{2s}}{2s+1} {\cal P} X\,, \qquad ( P^{-} )^{2} =  P^{-} \,.
\ee
Hence, we can use
the fusion procedure \cite{Kulish:1981gi, Kulish:1981bi}, as 
generalized to the case of boundaries in \cite{Mezincescu:1991ke}, to 
obtain the fusion formula
\be
t(u; \{\theta_{j}\})\, t(u q; \{\theta_{j}\}) = \frac{1}{\zeta(u^{2} 
q^{2})} \left[\tilde t(u; 
\{\theta_{j}\}) + f(u)\, \id^{\otimes N} \right]\,,
\label{fusion}
\ee
where $\tilde t(u; \{\theta_{j}\})$ is a fused transfer matrix, and 
the scalar function $f(u)$ is given by a product of quantum 
determinants
\be
f(u) &=& \Delta(K^{+})\, \Delta(K^{-})\, \delta(T(u))\, \delta(\hat 
T(u)) \non \\
&=& g(u^{-2} q^{-3})\, g(u^{2} q)\, \prod_{i=1}^{N}\left[ \zeta( u 
q/\theta_{i})\, \zeta( u q \theta_{i}) \right] 
\,,
\ee
where $g(u)$ is given by
\be
g(u) = \tr_{12} R_{12}(u) V_{1} V_{2} P_{12}^{-} = (-1)^{2s+1} 
\omega(u q^{-1}) \,.
\ee
Using the fact that the fused transfer matrix vanishes when evaluated 
at $q^{-1}\theta_{i}$, i.e.
\be
\tilde t(q^{-1}\theta_{i}; \{\theta_{j}\})=0\,, \qquad i = 1, \ldots, N\,,
\ee
it follows from (\ref{fusion}) that the fundamental transfer matrix 
(\ref{transferinhomo}) satisfies a set of exact functional relations
\be
t(q^{-1}\theta_{i}; \{\theta_{j}\})\, t(\theta_{i}; \{\theta_{j}\}) = 
F(q^{-1}\theta_{i})\, \id^{\otimes N}\,, \qquad i = 1, \ldots, N\,,
\label{funcrltntransf}
\ee
where 
\be
F(u) = \frac{f(u)}{\zeta(u^{2} q^{2})} = -\frac{\omega(u^{2})\, 
\omega(u^{2} q^{4})}{\omega(u^{2} q)\, \omega(u^{-2}q^{-3})} 
\prod_{i=1}^{N}\left[ \omega( u/\theta_{i})\, \omega( u 
q^{2}/\theta_{i})\, \omega(u \theta_{i})\, \omega( u q^{2} \theta_{i})  
\right] \,.
\label{functionF}
\ee 

Let us denote the eigenvalues of $t(u; \{\theta_{j}\})$ by 
$\Lambda(u; \{\theta_{j}\})$. From (\ref{crossingtransf}) it follows 
that the eigenvalues have crossing symmetry
\be
\Lambda(u; \{\theta_{j}\}) = \Lambda(-u^{-1} q^{-1}; 
\{\theta_{j}\})\,;
\label{crossing}
\ee
and from (\ref{funcrltntransf}) it follows 
that the eigenvalues obey the functional relations
\be
\Lambda(q^{-1}\theta_{i}; \{\theta_{j}\})\, \Lambda(\theta_{i}; \{\theta_{j}\}) = 
F(q^{-1}\theta_{i})\,, \qquad i = 1, \ldots, N\,.
\label{funcrltn}
\ee
To solve these equations, we introduce the functions
\be
a(u; \{\theta_{j}\}) &=&-\frac{\omega(u^{2} q^{2})}{\omega(u^{2} q)} \prod_{i=1}^{N}\left[ 
\omega( u q/\theta_{i})\, \omega(u q \theta_{i}) \right]\,, \non\\
d(u; \{\theta_{j}\}) &=&-\frac{\omega(u^{2})}{\omega(u^{2} q)} \prod_{i=1}^{N}\left[ 
\omega( u/\theta_{i})\, \omega(u\theta_{i}) \right] = a(-u^{-1} q^{-1}; \{\theta_{j}\})\,,
\ee
which have the properties 
\be
a(q^{-1}\theta_{i}; \{\theta_{j}\}) = 0 = d(\theta_{i}; \{\theta_{j}\})\,, 
\qquad a(\theta_{i}; \{\theta_{j}\})\, 
d(q^{-1}\theta_{i}; \{\theta_{j}\}) = F(q^{-1}\theta_{i})\,.
\ee 
It is easy to see that equations (\ref{crossing}) and (\ref{funcrltn}) are satisfied by
\be
\Lambda(u; \{\theta_{j}\}) = a(u; \{\theta_{j}\})\, \frac{\mQ(u q^{-1})}{\mQ(u)} 
+ d(u; \{\theta_{j}\})\, \frac{\mQ(u q)}{\mQ(u)} \,, 
\label{lambdagen}
\ee
where $\mQ(u)$ is any crossing-invariant function
\be
\mQ(u) = \mQ(-u^{-1} q^{-1}) \,.
\label{Qcrossing}
\ee
From the form (\ref{Rmatrix}) of the R-matrix and the commutativity 
property (\ref{commutativity}), it follows that $\Lambda(u; 
\{\theta_{j}\})$ must be a Laurent polynomial in $u$ (with a finite 
number of terms). We assume that $\mQ(u)$ is also a Laurent 
polynomial, and is given by
\be
\mQ(u) = \prod_{k=1}^{M}\omega(u/u_{k})\, \omega(u q u_{k})\,,
\label{Qfunc}
\ee
where the so-called Bethe roots
$\{ u_{1}, \ldots, u_{M} \}$ are still to be determined,
which is consistent with (\ref{Qcrossing}). Obviously $\mQ(u_{k}) = 
0$, which means that both terms in the expression (\ref{lambdagen}) for $\Lambda(u; 
\{\theta_{j}\})$ have a simple pole at $u=u_{k}$. (We assume that the 
Bethe roots are distinct, and are not equal to 0 or $\infty$.) The corresponding residues 
must cancel (since $\Lambda(u; \{\theta_{j}\})$ must be finite for $u$ not equal to 0 or $\infty$), 
which implies the so-called Bethe equations
\be
\frac{a(u_{k}; \{\theta_{j}\})}{d(u_{k}; \{\theta_{j}\})}
= -\frac{\mQ(u_{k} q)}{\mQ(u_{k} q^{-1})}\,, \qquad k=1, \ldots, M\,.
\ee
Since we no longer need the inhomogeneities, we now set them to unity 
$\theta_{j}=1$. 

To summarize, we have argued that the eigenvalues 
$\Lambda(u)$ of the transfer matrix $t(u)$ (\ref{transfer}) are given 
by 
\be
\Lambda(u) &=& -\frac{1}{\omega(u^{2} q)}\Bigg[ \omega(u^{2} q^{2})\, \omega(u q)^{2N}
\prod_{j=1}^{M}\frac{\omega(u q^{-1}/u_{j})\, \omega(u 
u_{j})}{\omega(u/u_{j})\, \omega(u q u_{j})} \non\\
&&+ \omega(u^{2})\, \omega(u)^{2N}
\prod_{j=1}^{M}\frac{\omega(u q/u_{j})\, \omega(u q^{2} u_{j})}{\omega(u/u_{j})\, \omega(u q u_{j})}
\Bigg] \,, 
\label{lambda}
\ee
and the Bethe roots are given by the Bethe equations
\be
\left[\frac{\omega(u_{k} q)}{\omega(u_{k})}\right]^{2N} = 
\prod_{\scriptstyle{j \ne k}\atop \scriptstyle{j=1}}^M
\frac{\omega(u_{k} u_{j}^{-1} q)\, \omega(u_{k} u_{j} q^{2})}
{\omega(u_{k} u_{j}^{-1} q^{-1})\, \omega(u_{k} u_{j})} \,.
\label{BE}
\ee
These equations take a more symmetric form in terms of the rescaled 
Bethe roots $\tilde u_{k} \equiv u_{k} q^{1/2}$:
\be
\left[\frac{\omega(\tilde u_{k} q^{1/2})}{\omega(\tilde u_{k} q^{-1/2})}\right]^{2N} = 
\prod_{\scriptstyle{j \ne k}\atop \scriptstyle{j=1}}^M
\frac{\omega(\tilde u_{k} \tilde u_{j}^{-1} q)\, \omega(\tilde u_{k} \tilde u_{j} q)}
{\omega(\tilde u_{k} \tilde u_{j}^{-1} q^{-1})\, \omega(\tilde u_{k} 
\tilde u_{j} q^{-1})} \,.
\label{BEsym}
\ee
Note that the results (\ref{lambda}) - (\ref{BEsym}) do not depend 
on the value of $s$; in particular, they coincide with the well-known results for the
case $s=1/2$ \cite{Pasquier:1989kd, Mezincescu:1991ag}. This is 
consistent with the TL equivalence (see e.g. 
\cite{Temperley:1971iq, Barber:1989zz, Baxter:1982zz, Baxter:1982xp, Martin}), which suggests
that the spectrum (but not the degeneracies) of the 
TL Hamiltonian (\ref{Hamiltonian}) is independent of the 
representation.

We have verified numerically for small values of $N$ and $s$ that 
every distinct eigenvalue of the transfer matrix can be expressed 
in the form (\ref{lambda}). See Tables \ref{table:N2} -  
\ref{table:N4}, and note that all $(2s+1)^{N}$ eigenvalues are 
accounted for.

\begin{table}[h!]
   \small
 \centering
 \begin{tabular}{cc|c|c|c|}
   \cline{3-5}
   & & \multicolumn{3}{ c| }{Degeneracies}\\
   \hline
   \multicolumn{1}{ |c|  } {$M$} & {  $\{ u_{k} \}$  } & $s=\frac{1}{2}$ & $s=1$  & $s=\frac{3}{2}$\\
   \hline
    \multicolumn{1}{ |c|  } 0 & - & 3 & 8 & 15 \\
    \multicolumn{1}{ |c|  } 1 & $1.34164 + 0.447214 i$ & 1 & 1 & 1 \\
   \hline
   & \multicolumn{1}{ c|  } {total:} & 4 & 9 & 16\\
   \cline{3-5}
   \end{tabular}
  \caption{\small Solutions $\{ u_{k} \}$ of the Bethe equations (\ref{BE}) and degeneracies of the corresponding
  eigenvalues (\ref{lambda}) for $N=2$ and $s=\frac{1}{2}, 1, 
  \frac{3}{2}$ with $q = 0.5$.}
  \label{table:N2}
\end{table}

\begin{table}[h!]
   \small
 \centering
 \begin{tabular}{cc|c|c|c|}
   \cline{3-5}
   & & \multicolumn{3}{ c| }{Degeneracies}\\
   \hline
   \multicolumn{1}{ |c|  } {$M$} & {  $\{ u_{k} \}$  } & $s=\frac{1}{2}$ & $s=1$  & $s=\frac{3}{2}$\\
   \hline
    \multicolumn{1}{ |c|  } 0 & - & 4 & 21 & 56 \\
    \multicolumn{1}{ |c|  } 1 & $1.22474 + 0.707107 i$ & 2 & 3 & 4 \\
    \multicolumn{1}{ |c|  } 1 & $1.38873 + 0.267261 i$ & 2 & 3 & 4 \\
   \hline
   & \multicolumn{1}{ c|  } {total:} & 8 & 27 & 64\\
   \cline{3-5}
   \end{tabular}
  \caption{\small Solutions $\{ u_{k} \}$ of the Bethe equations (\ref{BE}) and degeneracies of the corresponding
  eigenvalues (\ref{lambda}) for $N=3$ and $s=\frac{1}{2}, 1, 
  \frac{3}{2}$ with $q = 0.5$.}
  \label{table:N3}
\end{table}

\begin{table}[h!]
   \small
 \centering
 \begin{tabular}{cc|c|c|c|}
   \cline{3-5}
   & & \multicolumn{3}{ c| }{Degeneracies}\\
   \hline
   \multicolumn{1}{ |c|  } {$M$} & {  $\{ u_{k} \}$  } & $s=\frac{1}{2}$ & $s=1$  & $s=\frac{3}{2}$\\
   \hline
    \multicolumn{1}{ |c|  } 0 & - & 5 & 55 & 209 \\
    \multicolumn{1}{ |c|  } 1 & $1.10176 + 0.886631 i$ & 3 & 8 & 15 \\
    \multicolumn{1}{ |c|  } 1 & $1.34164 + 0.447214 i$ & 3 & 8 & 15 \\
    \multicolumn{1}{ |c|  } 1 & $1.40092 + 0.193427 i$ & 3 & 8 & 15 \\
    \multicolumn{1}{ |c|  } 2 & $1.81555 - 
    0.854196 i$, $1.81555 + 
    0.854196 i$ & 1 & 1 & 1 \\
    \multicolumn{1}{ |c|  } 2 & $1.28401 + 0.592723 i$, $1.3969 + 
    0.220635 i$ & 1 & 1 & 1 \\
   \hline
   & \multicolumn{1}{ c|  } {total:} & 16 & 81 & 256\\
   \cline{3-5}
   \end{tabular}
  \caption{\small Solutions $\{ u_{k} \}$ of the Bethe equations (\ref{BE}) and degeneracies of the corresponding
  eigenvalues (\ref{lambda}) for $N=4$ and $s=\frac{1}{2}, 1, 
  \frac{3}{2}$ with $q = 0.5$.}
  \label{table:N4}
\end{table}

In closing this section, we note that the eigenvalues of the 
Hamiltonian (\ref{Hamiltonian}) are given by
\be
E = \alpha \frac{d}{du}\Lambda(u)\Big\vert_{u=1} + \beta = 
\frac{1}{2}\omega(q)\sum_{j=1}^{M}\left[\frac{\omega(u_{j}^{2})}{\omega(u_{j})^{2}}
-\frac{\omega(u_{j}^{2} q^{2})}{\omega(u_{j} q)^{2}}\right]\,,
\ee
as follows from (\ref{Htransfrltn}),  (\ref{alphabeta}) and  
(\ref{lambda}) .

\section{Quantum group symmetry}\label{sec:symmetry}

In this section we demonstrate the quantum-group invariance of the transfer 
matrix, and we discuss the implications of this symmetry for the 
Bethe ansatz solution.

\subsection{Symmetry of the transfer matrix}\label{sec:symmtransf}

Let us denote by $R^{\pm}$ the asymptotic limits of the R-matrix $R(u)$
(\ref{Rmatrix})
\be
R^{+} &=& \lim_{u\rightarrow\infty}\frac{1}{u}\, R(u) = {\cal P}(q+X) \,,
\non \\
R^{-} &=& \lim_{u\rightarrow 0}-u\, R(u) = {\cal P}(q^{-1}+X) \,,
\ee
and let us similarly denote by $T^{\pm}_{0}$ the asymptotic limits of the monodromy matrix 
$T_{0}(u)$ (\ref{monodromy})
\be
T^{\pm}_{0} = R^{\pm}_{0N} \cdots R^{\pm}_{01} \,.
\ee
Regarding $T^{\pm}_{0}$ as a $(2s+1) \times (2s+1)$ matrix in the
auxiliary space, its matrix elements $T^{\pm}_{ij}$ (which are
operators on the quantum space $\left({\mathbb C}^{2s+1}\right)^{\otimes
N}$) define a quantum group, which has been identified in \cite{KMN}
as $U_{q}(2s+1)$.\footnote{For $s>1/2$, this symmetry is larger 
than the $U_{Q}sl(2)$ symmetry \cite{Batchelor:1989uk}, which is also present.}
We shall demonstrate that each of these matrix
elements commutes with the transfer matrix
\be
\left[T_{ij}^{\pm}\,, t(u) \right] = 0 \,, \qquad i, j = 1, 2, \ldots, 
2s+1\,.
\label{qgsymmetry}
\ee
The proof, which is similar to the 
one in \cite{Mezincescu:1991rb} (see also  \cite{Kulish:1991np}), 
requires two lemmas:

\begin{lemma}
\be
\left[R_{12}^{\pm}\, T_{1}^{\pm}\,, T_{2}(u)\, \hat T_{2}(u) 
\right] = 0 \,.
\label{lemma1}
\ee
\end{lemma}
\begin{proof}
\quad We recall the fundamental relation
\be
R_{12}(u_{1}/u_{2})\, T_{1}(u_{1})\, T_{2}(u_{2}) = T_{2}(u_{2})\, 
T_{1}(u_{1})\, R_{12}(u_{1}/u_{2}) \,.
\ee
Taking asymptotic limits of $u_{1}$ yields
\be
R_{12}^{\pm}\, T_{1}^{\pm}\, T_{2}(u) = T_{2}(u)\, 
T_{1}^{\pm}\, R_{12}^{\pm} \,,
\label{RTTasym}
\ee
which further implies
\be
T_{2}^{-1}(u)\, R_{12}^{\pm}\,  T_{1}^{\pm} = T_{1}^{\pm}\, 
R_{12}^{\pm}\, T_{2}^{-1}(u) \,.
\label{TRTasym}
\ee
Therefore,
\be
R_{12}^{\pm}\, T_{1}^{\pm}\, T_{2}(u)\, T_{2}^{-1}(u^{-1})  &=& 
T_{2}(u)\, T_{1}^{\pm}\, R_{12}^{\pm} \, T_{2}^{-1}(u^{-1}) \non\\
&=& T_{2}(u)\, T_{2}^{-1}(u^{-1})\, R_{12}^{\pm}\,  T_{1}^{\pm} \,,
\ee
where the first equality follows from (\ref{RTTasym}), and the second 
equality follows from (\ref{TRTasym}). We have therefore shown the 
commutativity property
\be
\left[R_{12}^{\pm}\, T_{1}^{\pm}\,, T_{2}(u)\, T_{2}^{-1}(u^{-1}) 
\right] = 0 \,.
\label{comm}
\ee
Furthermore,
\be
T_{0}^{-1}(u) &=& R_{01}^{-1}(u) \cdots R_{0N}^{-1}(u) \non \\
&\propto& R_{10}(u^{-1}) \cdots R_{N0}(u^{-1}) = \hat 
T_{0}(u^{-1})\,,
\ee
where the second line follows from unitarity (\ref{unitarity}).
Substituting into (\ref{comm}) we obtain the desired result 
(\ref{lemma1}). 
\end{proof}

\begin{lemma}
\be
M^{-1}_{1}\, 
\left((R^{\pm}_{12})^{-1}\right)^{t_{2}}\, M_{1}\, R^{\pm\ 
t_{2}}_{12} = \id^{\otimes 2} \,.
\label{lemma2}
\ee
\end{lemma}
\begin{proof}
\quad We write the unitarity condition 
(\ref{unitarity}) as
\be
R_{12}(u)\, R_{12}^{t_{1} t_{2}}(u^{-1}) = \zeta(u)\, \id^{\otimes 
2}\,,
\label{unitarity2}
\ee
and then use crossing symmetry (\ref{crossingR}) to obtain
\be
V_{1}\, R_{12}^{t_{2}}(-u^{-1} q^{-1})\, V_{1}\, V_{1}^{t_{1}}\, 
R_{12}^{t_{1}}(-u q^{-1})\, V_{1}^{t_{1}} = \zeta(u)\, \id^{\otimes 2} 
\,.
\ee
By taking asymptotic limits and noting that 
$V^{2}=(-1)^{2s}\, \id$, we obtain
\be
R^{\pm\ t_{2}}_{12}\, M^{-1}_{1}\, R^{\mp\ t_{1}}_{12}\, M_{1} = 
\id^{\otimes 2}  \,.
\label{intermed}
\ee
Moreover, from (\ref{unitarity2}) we obtain 
$R^{\pm}_{12}\, R_{12}^{\mp\ t_{1} t_{2}} = \id$, which implies that 
\be
R_{12}^{\mp\ t_{1} t_{2}} = (R^{\pm}_{12})^{-1} \,, \quad \mbox{  or  } 
\quad
R_{12}^{\mp\ t_{1}} = \left((R^{\pm}_{12})^{-1}\right)^{t_{2}}  \,.
\ee
Substituting into (\ref{intermed}), we obtain
\be
R^{\pm\ t_{2}}_{12}\, M^{-1}_{1}\, \left((R^{\pm}_{12})^{-1}\right)^{t_{2}}\, M_{1} = 
\id^{\otimes 2}  \,,
\ee
which can be rearranged to give the desired result (\ref{lemma2}).  
\end{proof}

We are now ready to prove the main result (\ref{qgsymmetry}), which 
is equivalent to the following

\begin{prop*}
\be
\left[T_{1}^{\pm}\,, t(u) \right] = 0 \,.
\label{prop}
\ee
\end{prop*}
\begin{proof}
\quad Recalling the expression (\ref{transfer}) for the transfer matrix, we 
obtain
\be
T_{1}^{\pm}\, t(u) &=& \tr_{2}\left\{T_{1}^{\pm}\, M_{2}\, T_{2}(u)\, 
\hat T_{2}(u) \right\} \non \\
&=& \tr_{2}\left\{M^{-1}_{1}\, M_{1}\,M_{2}\, (R^{\pm}_{12})^{-1}\, 
R^{\pm}_{12}\, T_{1}^{\pm}\,  T_{2}(u)\, \hat T_{2}(u) \right\} \non \\
&=& \tr_{2}\left\{M^{-1}_{1}\, (R^{\pm}_{12})^{-1}\, M_{1}\,M_{2}\,  
T_{2}(u)\, \hat T_{2}(u)\, R^{\pm}_{12}\, T_{1}^{\pm}   \right\} = 
\ldots 
\ee
In passing to the third line, we have used the fact 
$\left[ M_{1}\,M_{2}\,, R^{\pm}_{12} \right] = 0$ as well as the first
lemma (\ref{lemma1}). Then
\be
\ldots &=& \tr_{2}\left\{M^{-1}_{1}\, (R^{\pm}_{12})^{-1}\, M_{1}\,M_{2}\,  
T_{2}(u)\, \hat T_{2}(u)\, R^{\pm}_{12} \right\} T_{1}^{\pm} 
\non \\
&=& \tr_{2}\left\{ A_{12}\, Z_{2}\, R^{\pm}_{12}\, \right\} T_{1}^{\pm}  \non \\
&=& \tr_{2}\left\{ A_{12}^{t_{2}}\, R^{\pm\ t_{2}}_{12}\, Z_{2}^{t_{2}}  \right\} T_{1}^{\pm} = \ldots  
\ee 
In passing to the second line we have made the identifications
$A_{12} = M^{-1}_{1}\, (R^{\pm}_{12})^{-1}\, M_{1}$
and $Z_{2} = M_{2}\, T_{2}(u)\, \hat T_{2}(u)$. Finally, we obtain
\be
\ldots &=& \tr_{2}\left\{ M^{-1}_{1}\, 
\left((R^{\pm}_{12})^{-1}\right)^{t_{2}}\, M_{1}\, R^{\pm\ 
t_{2}}_{12}\,  Z_{2}^{t_{2}} \right\} T_{1}^{\pm} \non \\
&=&  \tr_{2}\left\{ Z_{2}^{t_{2}} \right\} T_{1}^{\pm} \non \\
&=& t(u)\,  T_{1}^{\pm} \,.
\ee
In passing to the second line we have used the second lemma 
(\ref{lemma2}).  
\end{proof}

\subsection{Degeneracies and multiplicities}\label{sec:degmult}

The $U_{q}(2s+1)$ symmetry of the transfer matrix implies that its 
eigenstates form representations of this algebra. The space of states 
has the decomposition (see e.g. \cite{Martin,Kulish,Benkart2005,KMN,Aufgebauer:2010gg} 
and references therein)
\be
\left({\mathbb C}^{2s+1}\right)^{\otimes N} = \bigoplus_{k=0(1)}^{N} 
\nu_{k}\, V_{k}\,,
\label{decomp}
\ee
where the summation is over even (odd) integers for even (odd) $N$, 
respectively; $V_{k}$ are representations of $U_{q}(2s+1)$; and 
$\nu_{k}$ are the multiplicities. The dimensions of the 
representations are given by
\be
\dim V_{k} = p_{k}(2s+1)\,,
\label{dimension}
\ee
where $p_{k}(x)$ are Chebyshev polynomials of the second kind, which 
are defined by the recurrence relations
\be
p_{k+1}(x) + p_{k-1}(x) = x\, p_{k}(x)\,, \qquad p_{0}(x) = 1\,, \qquad 
p_{-1}(x) = 0\,.
\ee
The multiplicities are given by \footnote{Note that the 
multiplicities $\nu_{k}$ are independent of $s$.}
\be
\nu_{k} = \left\{\begin{array}{ll}
{N\choose \frac{N-k}{2}} - {N\choose \frac{N-k}{2}-1}  
& \qquad k \ne 0\,, N  \\
\frac{1}{\frac{N}{2}+1}{N\choose \frac{N}{2}} & \qquad k=0 \\
1 & \qquad k=N
\end{array}\right. \,,
\label{multiplicity}
\ee
which are the dimensions of representations $W_{k}$ of the TL algebra 
$TL_{N}$. As a check on (\ref{decomp}) - (\ref{multiplicity}),
one can verify that the sum rule
\be
\sum_{k=0(1)}^{N} \nu_{k}\, \dim V_{k} = (2s+1)^{N}
\ee
is satisfied.

For given values of $N$ and $s$, let ${\cal N}(N, M)$ denote the
number of solutions of the Bethe equations (\ref{BE}) with $M$ roots,
and let ${\cal D}(N, s, M)$ denote the corresponding degeneracy, i.e.,
the number of transfer-matrix eigenvalues (\ref{lambda})
corresponding to each solution of the Bethe equations with $M$ roots. 
We propose that ${\cal N}(N, M)$ and ${\cal D}(N, s, M)$ are related 
to  $\nu_{k}$ and $\dim V_{k}$ in the following simple way: 
\footnote{We note that \cite{Aufgebauer:2010gg} does not discuss 
either the open-chain Bethe equations (\ref{BE}) or the open-chain transfer matrix 
(\ref{transfer}), and therefore does not contain the results 
(\ref{numbersltns})-(\ref{Mkreltn}).}
\be
{\cal N}(N, M) &=& \nu_{k} \,, 
\label{numbersltns}\\
{\cal D}(N, s, M) &=& \dim V_{k}\,, 
\label{degeneracy}
\ee
with
\be
M=\frac{1}{2}(N-k)\,.
\label{Mkreltn}
\ee
We have verified these relations for small values of $N$ and $s$.  
See e.g. Tables \ref{table:N2gptheory} - \ref{table:N4gptheory}, and compare
with Tables \ref{table:N2} - \ref{table:N4}, respectively.

\begin{table}[h!]
   \small
 \centering
 \begin{tabular}{cc|c|c|c|}
   \cline{3-5}
   & & \multicolumn{3}{ c| }{$\dim V_{k}$}\\
   \hline
   \multicolumn{1}{ |c|  } {$k$} & {  $\nu_{k}$  } & $s=\frac{1}{2}$ & $s=1$  & $s=\frac{3}{2}$\\
   \hline
    \multicolumn{1}{ |c|  } 0 & 1 & 1 & 1 & 1 \\
    \multicolumn{1}{ |c|  } 2 & 1 & 3 & 8 & 15 \\
    \hline
    \end{tabular}
  \caption{\small Dimensions  (\ref{dimension}) and
  multiplicities $\nu_{k}$ (\ref{multiplicity}) of representations $V_{k}$ for $N=2$ and $s=\frac{1}{2}, 1, 
  \frac{3}{2}$.}
  \label{table:N2gptheory}
\end{table}

\begin{table}[h!]
   \small
 \centering
 \begin{tabular}{cc|c|c|c|}
   \cline{3-5}
   & & \multicolumn{3}{ c| }{$\dim V_{k}$}\\
   \hline
   \multicolumn{1}{ |c|  } {$k$} & {  $\nu_{k}$  } & $s=\frac{1}{2}$ & $s=1$  & $s=\frac{3}{2}$\\
   \hline
    \multicolumn{1}{ |c|  } 1 & 2 & 2 & 3 & 4 \\
    \multicolumn{1}{ |c|  } 3 & 1 & 4 & 21 & 56 \\
    \hline
    \end{tabular}
  \caption{\small Dimensions  (\ref{dimension}) and
  multiplicities $\nu_{k}$ (\ref{multiplicity}) of representations 
  $V_{k}$ for $N=3$ and $s=\frac{1}{2}, 1, 
  \frac{3}{2}$.}
  \label{table:N3gptheory}
\end{table}

\begin{table}[h!]
   \small
 \centering
 \begin{tabular}{cc|c|c|c|}
   \cline{3-5}
   & & \multicolumn{3}{ c| }{$\dim V_{k}$}\\
   \hline
   \multicolumn{1}{ |c|  } {$k$} & {  $\nu_{k}$  } & $s=\frac{1}{2}$ & $s=1$  & $s=\frac{3}{2}$\\
   \hline
    \multicolumn{1}{ |c|  } 0 & 2 & 1 & 1 & 1 \\
    \multicolumn{1}{ |c|  } 2 & 3 & 3 & 8 & 15 \\
    \multicolumn{1}{ |c|  } 4 & 1 & 5 & 55 & 209 \\
    \hline
    \end{tabular}
  \caption{\small Dimensions  (\ref{dimension}) and
  multiplicities $\nu_{k}$ (\ref{multiplicity}) of representations 
  $V_{k}$ for $N=4$ and $s=\frac{1}{2}, 1, 
  \frac{3}{2}$.}
  \label{table:N4gptheory}
\end{table}

\section{Algebraic Bethe ansatz}\label{sec:ABA}

We present here several conjectures related to the algebraic Bethe 
ansatz solution of the TL chain. The conjecture for the off-shell 
equation has been proved for $s=1$ \cite{Nepomechie:2016ruv}, while the 
other conjectures have been checked numerically (up to $M=3$, $N=6$ and $s=\frac{3}{2}$).

\subsection{Off-shell equation}
In order to implement the algebraic Bethe ansatz, we need to choose a convenient
representation in the auxiliary space for the double-row monodromy matrix $T_0(u)\,\hat T_0(u)$. We choose
\be
T_0(u)\,\hat T_0(u)=\left(
\begin{array}{ccccc}
\mA(u) & \mB_{1,2}(u)&\cdots&\mB_{1,2s}(u) & \mB(u) \\
\mC_{2,1}(u) & \mA_2(u) &\cdots&\mB_{2,2s}(u)& \mB_{2,2s+1}(u) \\
\vdots & \vdots &\ddots&\vdots& \vdots \\
\mC_{2s,1}(u) & \mC_{2s,2}(u) &\cdots&\mA_{2s}(u)& \mB_{2s,2s+1}(u) \\
\mC(u) & \mC_{2s+1,2}(u)&\cdots&\mC_{2s+1,2s}(u) & \tilde\mD(u)
\end{array}
\right)_{(2s+1)\times (2s+1)}\,,
\label{doublerowrep}
\ee
where each entry acts on the quantum space $\left({\mathbb C}^{2s+1}\right)^{\otimes
N}$. Let us also introduce the $(2s+1)-$dimensional reference state
\be
\refs=
{\left(
\begin{array}{c}1\\ 0 \\ 
\vdots\\0
\end{array}\right)}^{\otimes N}
\ee
and its dual
\be
\drefs={\left(
\begin{array}{cccc}1 & 0 & \cdots & 0
\end{array}\right)}^{\otimes N}
\ee
such that $\drefs0\rangle=1$. We have found that the (dual) Bethe vectors are generated
by the action of a single double-row operator, namely ($\mC(u)$) $\mB(u)$. Indeed,
let us define the Bethe vector as
\be\label{BV}
|u_1,\dots, u_M\rangle = \prod_{k=1}^M\mB(u_k)\refs
\ee
as well as its dual
\be\label{dBV}
\langle u_1,\dots,u_M| = \drefs\prod_{k=1}^M\mC(u_k)\,.
\ee
We conjecture that the action of the transfer matrix (\ref{transfer}) on (\ref{BV}) is given by
\be
t(u)|u_1,\dots ,u_M\rangle=\Lambda(u;u_1,\dots,u_M)|u_1\dots u_M\rangle
+\sum_{k=1}^M\lambda_k|u_1,\dots,u_{k-1},u,u_{k+1},\dots, u_M\rangle\,,
\label{rightoffshell}
\ee
while the action on (\ref{dBV}) is given by,
\be
\langle u_1,\dots,u_M|t(u)=\langle u_1,\dots,u_M|\Lambda(u;u_1,\dots,u_M)+\sum_{k=1}^M\langle u_1,\dots,u_{k-1},u,u_{k+1},\dots, u_M|\lambda_k\,,
\label{leftoffshell}
\ee
where
\be
\lefteqn{\lambda_k=-\frac{\omega(q)\omega(u^2q^2)\omega(u_k^2)}{\omega(uu_k^{-1})\omega(uu_kq)\omega(u_k^2q)}}\non\\
&&\times
\left[\omega(u_kq)^{2N}\prod_{\scriptstyle{j \ne k}\atop \scriptstyle{j=1}}^M\frac{\omega(u_ku_j^{-1}q^{-1})\omega(u_ku_j)}{\omega(u_ku_j^{-1})\omega(u_ku_jq)}
-\omega(u_k)^{2N}\prod_{\scriptstyle{j \ne k}\atop 
\scriptstyle{j=1}}^M\frac{\omega(u_ku_j^{-1}q)\omega(u_ku_jq^2)}{\omega(u_ku_j^{-1})\omega(u_ku_jq)}\right]\,,
\ee
and $\Lambda(u;u_1,\dots,u_M)$ is given by (\ref{lambda}).  Since the
equations (\ref{rightoffshell}) and (\ref{leftoffshell}) are valid for
arbitrary $\{u_k\}$, we write explicitly the dependence of $\Lambda$
on $\{u_k\}$.  Evidently, $\lambda_k=0$ when the Bethe equations
(\ref{BE}) are satisfied, in which case the Bethe states 
(\ref{BV}) and (\ref{dBV}) are right and left eigenstates of the 
transfer matrix $t(u)$, respectively, with corresponding eigenvalue $\Lambda(u;u_1,\dots,u_M)$.  
For the $s=\frac{1}{2}$ case the results
(\ref{rightoffshell}) and (\ref{leftoffshell}) are known
\cite{Sklyanin:1988yz}; for the $s=1$ case a proof will be reported in
a separate paper \cite{Nepomechie:2016ruv}.

\subsection{Highest-weight property}\label{sec:hw}

When the Bethe states (\ref{BV}) 
are on shell (i.e., when $\{u_{1}, \ldots, u_{M}\}$ satisfy the Bethe equations (\ref{BE})),
we conjecture that
\be
T_{ii}^{+} |u_1,\dots,u_M \rangle &=&  h_{i} |u_1,\dots,u_M \rangle\,, \qquad i = 1, 2, \ldots, 2s+1 
\,, \\
T_{ij}^{+} |u_1,\dots,u_M \rangle &=&   0\,, \qquad i > j \,,
\ee
where $T_{ij}^{+}$ are quantum group generators defined in section \ref{sec:symmtransf}. That 
is, on-shell Bethe states are highest-weight states of the quantum 
group, in the sense that they are eigenstates of the diagonal elements of 
$T^{+}$, and are annihilated by the lower triangular elements of 
$T^{+}$. This would help account for the observations 
in section \ref{sec:degmult} that the degeneracies and multiplicities 
are given by group theory.

\subsection{Scalar products}\label{sec:sp}
 
Let us suppose that $\{u_1,\dots,u_M\}$ are Bethe roots. We propose that the scalar product
between the on-shell state $\langle u_1,\dots,u_M|$ and an arbitrary off-shell state 
$|v_1,\dots ,v_M\rangle$ is given by
\be
\lefteqn{\langle u_1,\dots,u_M|v_1,\dots ,v_M\rangle}\non\\
&&=\left(\frac{1}{2Q^{2s}}\right)^M\,
\prod_{i=1}^M\frac{\omega(u_i)^{2N}u_i\,\omega(u_i^2)}{\omega(u_i^2q)\omega(v_i^2q^2)}
\prod_{j<i}^M\frac{\omega(u_iu_jq^2)}{\omega(u_iu_j)}
\frac{\textrm{Det}_M\left(\frac{\partial}{\partial u_i}\Lambda(v_j;u_1,\dots,u_M)\right)}
{\textrm{Det}_M\left(\frac{1}{\omega(v_iu_j^{-1})\omega(v_iu_jq)}\right)}\,.
\label{slavnov}
\ee
The formula (\ref{slavnov}) was proved in \cite{Kitanine:2007bi} for the $s=\frac{1}{2}$ case with (diagonal) boundary fields (see also \cite{Wang2002633} for the XXX chain).
Performing the limit $v_k\rightarrow u_k$, we obtain the square of the norm, namely,
\be
\lefteqn{\langle u_1,\dots,u_M|u_1,\dots ,u_M\rangle}\non\\
&&=\left(\frac{\omega(q)\omega(-q^2)}{Q^{2s}}\right)^M\,
\prod_{i=1}^M\omega(u_i)^{4N}\omega(u_i^2)^2
\prod_{j<i}^M\frac{\omega(u_iu_jq^2)}{\omega(u_ju_i^{-1})\omega(u_iu_j^{-1})\omega(u_iu_j)\omega(u_iu_j q)^2}
\non\\&&\qquad\qquad\times\textrm{Det}_M\left(G\right)\,,
\label{norm}
\ee
where $G$ is a $M\times M$ matrix with elements
\be
\lefteqn{G_{ij}=\frac{\prod_{k\neq i,j}^M\omega(u_j u_k^{-1}q)\omega(u_ju_kq^2)}{\omega(u_ju_i^{-1}q^{-1})\omega(u_iu_j)}
\Bigg[1-\delta_{i,j}+\delta_{i,j}
\frac{\omega(q)\omega(u_i^2)}{\omega(q^2)\omega(u_i^2q)^2}}
\non\\
&&\times
\Bigg(-\frac{2N\omega(q)}{\omega(u_i)\omega(u_iq)}+\omega(q^2)\sum_{k\neq i}^{M}\frac{1}{\omega(u_iq^{-1}u_k^{-1})\omega(u_iqu_k^{-1})}+
\frac{1}{\omega(u_iu_k)\omega(u_iu_kq^2)}\Bigg)\Bigg]\,.
\label{Gij}
\ee

\section{Discussion}\label{sec:discussion}

We have considered the TL open quantum spin chain associated with the
spin-$s$ representation of quantum-deformed $sl(2)$.  We have
constructed the transfer matrix (\ref{transfer}), and we have seen
that its eigenvalues (\ref{lambda}) and the corresponding Bethe
equations (\ref{BE}) do not depend on the value of the spin.  Due to the
quantum-group invariance of the transfer matrix (\ref{qgsymmetry}),
(\ref{prop}), the number of solutions of the Bethe equations
(\ref{numbersltns}) and the degeneracies of the transfer-matrix
eigenvalues (\ref{degeneracy}) can be inferred from group theory.

We have proposed an algebraic Bethe ansatz construction of the Bethe
vectors (\ref{BV}) and (\ref{dBV}), and the corresponding off-shell
equations (\ref{rightoffshell}) and (\ref{leftoffshell}),
respectively.  Remarkably, despite the fact that the auxiliary space
has dimension greater than 2 for $s>1/2$, a single creation operator
suffices to construct all the Bethe states - no nesting is needed.  We
have also proposed a determinant formula for the scalar products
between off-shell and on-shell Bethe states (\ref{slavnov}).
Remarkably, these results are also universal in the sense that they
depend on the value of the spin only through a constant factor.  It is
important to find proofs for these conjectures.  So far, we have been
able to prove only the off-shell equations for $s=1$
\cite{Nepomechie:2016ruv}.

We have seen that appropriate boundary conditions are necessary for
the TL model to have a universal solution.  Indeed, for periodic
boundary conditions, the solution (\ref{BEclosed}) is no longer
universal.  This solution has the unusual feature that it has a twist
that is ``dynamically'' generated (i.e., the twist is not a fixed
parameter of the model, as is typically the case).  An algebraic Bethe
ansatz solution for this model remains to be found.
It may be interesting to consider generalizations of the open TL chain
(\ref{Hamiltonian}) which are still integrable but have boundary terms
that break the quantum group symmetry \cite{LimaSantos:2010nw, Avan:2010mh}.

\section*{Acknowledgments}
The work of RN was supported in part by the National Science
Foundation under Grant PHY-1212337, and by a Cooper fellowship.
RP thanks the S\~ao Paulo Research Foundation (FAPESP),
grants \# 2014/00453-8 and \# 2014/20364-0, for financial support. We also acknowledge the support by FAPESP and the University of Miami under the SPRINT
grant \#2016/50023-5.

\appendix

\section{Closed TL chain}\label{sec:closed} 

Let $t(u)$ now denote the transfer matrix for the closed TL chain 
with periodic boundary conditions
\be
t(u) = \tr_{0} T_{0}(u) \,,
\label{transferclosed}
\ee
where the monodromy matrix $T_{0}(u)$ is given by (\ref{monodromy}).
To determine the eigenvalues of $t(u)$, we follow the 
same approach used in Section \ref{sec:BA} to analyze the open 
chain. Hence, we consider the inhomogeneous transfer matrix
\be
t(u; \{\theta_{j}\}) = \tr_{0} T_{0}(u; \{\theta_{j}\})\,,
\ee 
where $T_{0}(u; \{\theta_{j}\})$ is given by (\ref{monodromyinhomo}).
Using the fusion procedure, we arrive at the functional relations 
\be
t(q^{-1}\theta_{i}; \{\theta_{j}\})\, t(\theta_{i}; \{\theta_{j}\}) = 
F(q^{-1}\theta_{i})\, \id^{\otimes N}\,, \qquad i = 1, \ldots, N\,,
\label{funcrltntransfclosed}
\ee
where $F(u)$ is now given by (cf. (\ref{functionF}))
\be
F(u) = 
\prod_{i=1}^{N}\left[ (-1)^{2s}\omega(u/\theta_{i})\, \omega( u q^{2}/\theta_{i}) \right] \,.
\ee 
The corresponding eigenvalues $\Lambda(u; \{\theta_{j}\})$ therefore 
obey the same functional relations 
\be
\Lambda(q^{-1}\theta_{i}; \{\theta_{j}\})\, \Lambda(\theta_{i}; \{\theta_{j}\}) = 
F(q^{-1}\theta_{i})\,, \qquad i = 1, \ldots, N\,.
\label{funcrltnclosed}
\ee
To solve these equations, we introduce the functions
\be
a(u; \{\theta_{j}\}) = \kappa \prod_{i=1}^{N} 
(-1)^{s}\omega( u q/\theta_{i}) \,, \qquad 
d(u; \{\theta_{j}\}) = \frac{1}{\kappa} \prod_{i=1}^{N} 
(-1)^{s}\omega( u/\theta_{i})\, \,,
\ee
where the twist parameter $\kappa$ is still to be determined.
We observe that 
\be
a(q^{-1}\theta_{i}; \{\theta_{j}\}) = 0 = d(\theta_{i}; \{\theta_{j}\})\,, 
\qquad a(\theta_{i}; \{\theta_{j}\})\, 
d(q^{-1}\theta_{i}; \{\theta_{j}\}) = F(q^{-1}\theta_{i})\,.
\ee 
Hence, the functional relations (\ref{funcrltnclosed}) are satisfied by
\be
\Lambda(u; \{\theta_{j}\}) = a(u; \{\theta_{j}\})\, \frac{\mQ(u q^{-1})}{\mQ(u)} 
+ d(u; \{\theta_{j}\})\, \frac{\mQ(u q)}{\mQ(u)} \,, 
\ee
where $\mQ(u)$ is now given by
\be
\mQ(u) = \prod_{k=1}^{M}\omega(u/u_{k})\,.
\label{Qfuncclosed}
\ee
Setting the inhomogeneities to unity $\theta_{j}=1$, we conclude that 
the eigenvalues $\Lambda(u)$ of the closed transfer matrix 
(\ref{transferclosed}) are given by
\be
\Lambda(u) = \kappa\, (-1)^{s N}\, \omega(u q)^{N}
\prod_{j=1}^{M}\frac{\omega(u q^{-1}/u_{j})}{\omega(u/u_{j})}
+ \frac{1}{\kappa}\, (-1)^{s N}\, \omega(u)^{N}
\prod_{j=1}^{M}\frac{\omega(u q/u_{j})}{\omega(u/u_{j})} \,, 
\label{lambdaclosed}
\ee
and the Bethe roots are given by the Bethe equations
\be
\left[\frac{\omega(u_{k} q)}{\omega(u_{k})}\right]^{N} = \kappa^{-2}
\prod_{\scriptstyle{j \ne k}\atop \scriptstyle{j=1}}^M
\frac{\omega(u_{k} u_{j}^{-1} q)}
{\omega(u_{k} u_{j}^{-1} q^{-1})} \,.
\label{BEclosed}
\ee
A similar solution was proposed in \cite{Kulish}, except with a 
trivial twist (i.e., with $\kappa=1$). Such a twist is 
not expected, since the transfer matrix (\ref{transferclosed}) 
corresponds to periodic boundary conditions. Nevertheless, 
from numerical studies (see below), we 
find that a nontrivial twist ($\kappa \ne 1$) is 
necessary in order to obtain the complete set of eigenvalues from the 
Bethe ansatz solution. The presence of an effective twist was already  
noted in earlier work, see e.g.
\cite{Alcaraz:1992uq, Aufgebauer:2010gg, Finch:2014nxa, Finch:2015}.

We remark that the twist is characterized by an integer in 
${\mathbb Z}_{N}$. Indeed, we observe from (\ref{Rmatrix}) that $R(1)=\omega(q)\, 
{\cal P}$. Hence, from (\ref{transferclosed}) we obtain
\be
t(1)=\omega(q)^{N}\, U\,,
\label{t1}
\ee
where $U = {\cal P}_{12}\, {\cal P}_{23} \ldots {\cal P}_{N-1, N}$ is 
the one-site shift operator, which satisfies $U^{N}=\id^{\otimes N}$. From 
(\ref{lambdaclosed}) we have
\be
\Lambda(1) = \kappa\, (-1)^{s N}\, \omega(q)^{N}
\prod_{j=1}^{M}\frac{\omega(q u_{j})}{\omega(u_{j})}\,.
\label{lamb1}
\ee
It follows from (\ref{t1}) and  (\ref{lamb1}) that the eigenvalue of $U$ (which we 
also denote by $U$) is given by 
\be
U = \kappa\, (-1)^{s N}\,
\prod_{j=1}^{M}\frac{\omega(q u_{j})}{\omega(u_{j})} \,.
\ee 
Since $U^{N}=1$, we conclude that the twist $\kappa$ and the Bethe 
roots $\{ u_{j} \}$ must satisfy the following constraint
\be
\kappa =  \frac{e^{i 2\pi l/N}}{(-1)^{s N}}\prod_{j=1}^{M}\frac{\omega(u_{j})}{\omega(q 
u_{j})}\,, \qquad l = 0, 1, \ldots, N-1 \,.
\label{twistintger}
\ee
In particular, the twist is characterized by an integer $l \in {\mathbb Z}_{N}$

We have verified numerically for small values of $N$ and $s$ that 
every distinct eigenvalue of the transfer matrix 
(\ref{transferclosed}) can be expressed 
in the form (\ref{lambdaclosed}). See Tables \ref{table:N2closed} and
\ref{table:N3closed}, and note that all $(2s+1)^{N}$ eigenvalues are 
accounted for. Note also that, in contrast with the case of the open chain, the 
solutions of the closed-chain Bethe equations  (\ref{BEclosed}) are not universal: the Bethe roots 
depend on the value of the spin $s$ (cf. Tables 
\ref{table:N2}-\ref{table:N4}).

\begin{table}[h!]
   \small
 \centering
 \begin{tabular}{cccc||cccc||cccc|}
     \hline
   \multicolumn{4}{ |c|| }{$s=\frac{1}{2}$} &  \multicolumn{4}{ |c|| 
   }{$s=1$} & \multicolumn{4}{ |c| }{$s=\frac{3}{2}$}\\
   \hline
   \multicolumn{1}{ |c| } {$M$} &  \multicolumn{1}{ |c| } {$\{ u_{k}  
   \}$}  &  \multicolumn{1}{ |c| } {$\kappa$}   & \multicolumn{1}{ 
   |c|| } {${\cal D}$} &
   \multicolumn{1}{ |c| } {$M$} &  \multicolumn{1}{ |c| } {$\{ u_{k}  
   \}$}  &  \multicolumn{1}{ |c| } {$\kappa$}   & \multicolumn{1}{ 
   |c|| } {${\cal D}$} &
   \multicolumn{1}{ |c| } {$M$} &  \multicolumn{1}{ |c| } {$\{ u_{k}  
   \}$}  &  \multicolumn{1}{ |c| } {$\kappa$}   & \multicolumn{1}{ 
   |c| } {${\cal D}$}  \\
   \hline
   \multicolumn{1}{ |c| } {0} &  \multicolumn{1}{ |c| } {-}  &  
  \multicolumn{1}{ |c| } {-1}   & \multicolumn{1}{ 
   |c|| } {2} &
   \multicolumn{1}{ |c| } {0} &  \multicolumn{1}{ |c| } {-}  &  
   \multicolumn{1}{ |c| } {1}   & \multicolumn{1}{ 
   |c|| } {5} &
   \multicolumn{1}{ |c| } {0} &  \multicolumn{1}{ |c| } {-}  &  
   \multicolumn{1}{ |c| } {-1}   & \multicolumn{1}{ 
   |c| } {9}  \\ 
    \multicolumn{1}{ |c| } {1} &  \multicolumn{1}{ |c| } {$1.41421 i$}  &  
  \multicolumn{1}{ |c| } {1}   & \multicolumn{1}{ 
   |c|| } {1} &
   \multicolumn{1}{ |c| } {0} &  \multicolumn{1}{ |c| } {-}  &  
   \multicolumn{1}{ |c| } {-1}   & \multicolumn{1}{ 
   |c|| } {2} &
   \multicolumn{1}{ |c| } {0} &  \multicolumn{1}{ |c| } {-}  &  
   \multicolumn{1}{ |c| } {1}   & \multicolumn{1}{ 
   |c| } {5}  \\  
 \multicolumn{1}{ |c| } {1} &  \multicolumn{1}{ |c| } {1.41421}  &  
  \multicolumn{1}{ |c| } {1}   & \multicolumn{1}{ 
   |c|| } {1} &
   \multicolumn{1}{ |c| } {1} &  \multicolumn{1}{ |c| } {0.540182}  &  
   \multicolumn{1}{ |c| } {0.381966}   & \multicolumn{1}{ 
   |c|| } {1} &
   \multicolumn{1}{ |c| } {1} &  \multicolumn{1}{ |c| } {0.732051}  &  
   \multicolumn{1}{ |c| } {0.267949}   & \multicolumn{1}{ 
   |c| } {1}  \\  
 \multicolumn{1}{ |c| } {} &  \multicolumn{1}{ |c| } {}  &  
  \multicolumn{1}{ |c| } {}   & \multicolumn{1}{ 
   |c|| } {} &
   \multicolumn{1}{ |c| } {1} &  \multicolumn{1}{ |c| } {1.21699}  &  
   \multicolumn{1}{ |c| } {0.381966}   & \multicolumn{1}{ 
   |c|| } {1} &
   \multicolumn{1}{ |c| } {1} &  \multicolumn{1}{ |c| } {1.1638}  &  
   \multicolumn{1}{ |c| } {0.267949}   & \multicolumn{1}{ 
   |c| } {1}  \\  
    \hline
   & &  \multicolumn{1}{ c  } {total:} & 4 & & & & 9 & & & & 16 
     \end{tabular}
  \caption{\small Solutions  $\{ u_{k} \}$ of the Bethe equations 
  (\ref{BEclosed}), twist $\kappa$, and degeneracies ${\cal D}$ of the corresponding
  eigenvalues (\ref{lambdaclosed}) for $N=2$ and $s=\frac{1}{2}, 1, 
  \frac{3}{2}$ with $q = 0.5$.}
  \label{table:N2closed}
\end{table}

\begin{table}[h!]
   \small
 \centering
 \begin{tabular}{cccc||cccc||cccc|}
     \hline
   \multicolumn{4}{ |c|| }{$s=\frac{1}{2}$} &  \multicolumn{4}{ |c|| 
   }{$s=1$} & \multicolumn{4}{ |c| }{$s=\frac{3}{2}$}\\
   \hline
   \multicolumn{1}{ |c| } {$M$} &  \multicolumn{1}{ |c| } {$\{ u_{k}  
   \}$}  &  \multicolumn{1}{|c|} {$\kappa$}   & \multicolumn{1}{ 
   |c|| } {${\cal D}$} &
   \multicolumn{1}{ |c| } {$M$} &  \multicolumn{1}{ |c| } {$\{ u_{k}  
   \}$}  &  \multicolumn{1}{ |c| } {$\kappa$}   & \multicolumn{1}{ 
   |c|| } {${\cal D}$} &
   \multicolumn{1}{ |c| } {$M$} &  \multicolumn{1}{ |c| } {$\{ u_{k}  
   \}$}  &  \multicolumn{1}{ |c| } {$\kappa$}   & \multicolumn{1}{ 
   |c| } {${\cal D}$}  \\
   \hline
   \multicolumn{1}{ |c| } {0} &  \multicolumn{1}{ |c| } {-}  &  
  \multicolumn{1}{ |c| } {$i$}   & \multicolumn{1}{ 
   |c|| } {2} &
   \multicolumn{1}{ |c| } {0} &  \multicolumn{1}{ |c| } {-}  &  
   \multicolumn{1}{ |c| } {-1}   & \multicolumn{1}{ 
   |c|| } {8} &
   \multicolumn{1}{ |c| } {0} &  \multicolumn{1}{ |c| } {-}  &  
   \multicolumn{1}{ |c| } {-$i$}   & \multicolumn{1}{ 
   |c| } {20}  \\ 
    \multicolumn{1}{ |c| } {1} &  \multicolumn{1}{ |c| } {$i\sqrt{2}e^{-i \pi/3}$}  &  
  \multicolumn{1}{ |c| } {-$i$}   & \multicolumn{1}{ 
   |c|| } {2} &
   \multicolumn{1}{ |c| } {0} &  \multicolumn{1}{ |c| } {-}  &  
   \multicolumn{1}{ |c| } {$e^{i \pi/3}$}   & \multicolumn{1}{ 
   |c|| } {5} &
   \multicolumn{1}{ |c| } {0} &  \multicolumn{1}{ |c| } {-}  &  
   \multicolumn{1}{ |c| } {$i e^{i \pi/3}$}   & \multicolumn{1}{ 
   |c| } {16}  \\  
 \multicolumn{1}{ |c| } {1} &  \multicolumn{1}{ |c| } {$-i\sqrt{2}e^{i \pi/3}$}  &  
  \multicolumn{1}{ |c| } {-$i$}   & \multicolumn{1}{ 
   |c|| } {2} &
   \multicolumn{1}{ |c| } {0} &  \multicolumn{1}{ |c| } {-}  &  
   \multicolumn{1}{ |c| } {$e^{-i \pi/3}$}   & \multicolumn{1}{ 
   |c|| } {5} &
   \multicolumn{1}{ |c| } {0} &  \multicolumn{1}{ |c| } {-}  &  
   \multicolumn{1}{ |c| } {$i e^{-i \pi/3}$}   & \multicolumn{1}{ 
   |c| } {16}  \\  
 \multicolumn{1}{ |c| } {1} &  \multicolumn{1}{ |c| } {$\sqrt{2}$}  &  
  \multicolumn{1}{ |c| } {-$i$}   & \multicolumn{1}{ 
   |c|| } {2} &
   \multicolumn{1}{ |c| } {1} &  \multicolumn{1}{ |c| } {$i\sqrt{2}$}  &  
   \multicolumn{1}{ |c| } {-1}   & \multicolumn{1}{ 
   |c|| } {3} &
   \multicolumn{1}{ |c| } {1} &  \multicolumn{1}{ |c| } {$-i\sqrt{2}e^{i \pi/3}$ }  &  
   \multicolumn{1}{ |c| } {$i$}   & \multicolumn{1}{ 
   |c| } {4}  \\  
 \multicolumn{1}{ |c| } {} &  \multicolumn{1}{ |c| } {}  &  
  \multicolumn{1}{ |c| } {}   & \multicolumn{1}{ 
   |c|| } {} &
   \multicolumn{1}{ |c| } {1} &  \multicolumn{1}{ |c| } 
   {$(3\sqrt{3}+i)/\sqrt{14}$}  &  
   \multicolumn{1}{ |c| } {-1}   & \multicolumn{1}{ 
   |c|| } {3} &
   \multicolumn{1}{ |c| } {1} &  \multicolumn{1}{ |c| } {$i\sqrt{2}e^{-i \pi/3}$}  &  
   \multicolumn{1}{ |c| } {$i$}   & \multicolumn{1}{ 
   |c| } {4}  \\  
   \multicolumn{1}{ |c| } {} &  \multicolumn{1}{ |c| } {}  &  
 \multicolumn{1}{ |c| } {}   & \multicolumn{1}{ 
   |c|| } {} &
   \multicolumn{1}{ |c| } {1} &  \multicolumn{1}{ |c| } {$(3\sqrt{3}-i)/\sqrt{14}$}  &  
   \multicolumn{1}{ |c| } {-1}   & \multicolumn{1}{ 
   |c|| } {3} &
   \multicolumn{1}{ |c| } {1} &  \multicolumn{1}{ |c| } {$\sqrt{2}$}  &  
   \multicolumn{1}{ |c| } {$i$}   & \multicolumn{1}{ |c| } {4}  \\  
 \hline
   & &  \multicolumn{1}{ c  } {total:} & 8 & & & & 27 & & & & 64 
     \end{tabular}
   \caption{\small Solutions  $\{ u_{k} \}$ of the Bethe equations 
  (\ref{BEclosed}), twist $\kappa$, and degeneracies ${\cal D}$ of the corresponding
  eigenvalues (\ref{lambdaclosed}) for $N=3$ and $s=\frac{1}{2}, 1, 
  \frac{3}{2}$ with $q = 0.5$.}
  \label{table:N3closed}
\end{table}


\providecommand{\href}[2]{#2}\begingroup\raggedright\endgroup

\end{document}